\documentclass[sigconf]{acmart}

\AtBeginDocument{%
  }

\setcopyright{acmlicensed}
\copyrightyear{2018}
\acmYear{2018}
\acmDOI{XXXXXXX.XXXXXXX}
\acmConference[Conference acronym 'XX]{Make sure to enter the correct
  conference title from your rights confirmation email}{June 03--05,
  2018}{Woodstock, NY}
\acmISBN{978-1-4503-XXXX-X/2018/06}

\newcommand{\method}{Agent4POI}

\newcommand{\cmark}{\ding{51}}
\newcommand{\xmark}{\ding{55}}

\usepackage{amsthm}
\usepackage{amsmath}
\usepackage{amssymb}
\usepackage{pifont}
\usepackage{algorithm}
\usepackage{algpseudocode}  
\usepackage{multirow}
\usepackage{makecell}
\usepackage{xcolor}
\usepackage[inline]{enumitem}
\usepackage{microtype}

\newtheorem{proposition}{Proposition}   
\theoremstyle{definition}
\newtheorem{definition}{Definition}     
\newtheorem{remark}{Remark}             

\begin{document}

\title{Agent4POI: Agentic Context-Conditioned Affordance Reasoning for Multimodal Point-of-Interest Recommendation}

\author{Jinze Wang*}
\affiliation{%
  \institution{Tongji University}
  \institution{Swinburne University of Technology}
  \city{Melbourne}
  \country{Australia}
}
\email{jinzewang@swin.edu.au}

\author{Yangchen Zeng*}
\affiliation{%
  \institution{Southeast University}
  \city{Wuxi}
  \country{China}
}
\email{zengyangchen@foxmail.com}
\thanks{*Both authors contributed equally to this work.}
\thanks{$\dagger$Corresponding author: Tiehua Zhang (tiehuaz@tongji.edu.cn)}

\author{Tiehua Zhang$\dagger$}
\affiliation{%
  \institution{Tongji University}
  \city{Shanghai}
  \country{China}
}
\email{tiehuaz@tongji.edu.cn}

\author{Lu Zhang}
\affiliation{%
 \institution{Chengdu University of Information Technology}
 \city{Chengdu}
 \country{China}}
\email{zhang\_lu010@outlook.com}

\author{Yuze Liu}
\affiliation{%
  \institution{Swinburne University of Technology}
  \city{Melbourne}
  \country{Australia}}
  \email{yuzeliu@swin.edu.au}

\author{Yongchao Liu}
\affiliation{%
  \institution{Ant Group}
  \city{Beijing}
  \country{China}}
  \email{yongchao.ly@antgroup.com}

\author{Xingjun Ma}
\affiliation{%
  \institution{Fudan University}
  \city{Shanghai}
  \country{China}}
\email{xingjunma@fudan.edu.cn}

\author{Zhu Sun}
\affiliation{%
  \institution{Singapore University of Technology and Design}
  \country{Singapore}}
  \email{zhu_sun@sutd.edu.sg}

\renewcommand{\shortauthors}{Trovato et al.}

\begin{abstract}
We introduce \method{}, the first POI recommendation framework that generates
\emph{context-conditioned} multimodal representations at recommendation time,
rather than relying on static POI embeddings pre-computed independently of context.
Existing multimodal systems encode each POI once as a static embedding, a design
that precludes reasoning about why the same caf\'{e} affords solo work on
Monday but group celebration on Friday evening.
We formally prove that no pre-computed encoder can satisfy context-sensitive ranking
under standard bilinear scoring, motivating inference-time item-side representation.
\method{} inverts this computation: given a situational context, a four-phase LLM agent
generates dynamic, context-specific affordance queries (Phase~1) and executes a five-step
cross-modal chain-of-thought over image, review, and metadata evidence (Phase~2).
The resulting uncertainty-aware affordance representation is grounded in Gibsonian
affordance theory. These cross-modal verdicts form a structured, uncertainty-adjusted affordance representation (Phase~3),
which is aligned with user preferences via a semantic caching system for low-latency ranking (Phase~4).
On three POI benchmarks and three evaluation configurations
(standard, cold-start, context-shift), \method{} achieves a \textbf{23.2\%}
relative gain over the strongest baseline and degrades by only \textbf{7.5\%}
under context-shift versus \textbf{16--17\%} for the strongest baselines.
%
In cold-start scenarios,  \method{} outperforms the best content-based baseline by up to $\mathbf{2.4{\times}}$, whereas ID-based methods fail to generalize.
\end{abstract}

\ccsdesc[500]{Information systems~Recommender systems}

\keywords{Point-of-Interest Recommendation; Multimodal Reasoning; Affordance Theory; LLM Agent; Cold-Start Recommendation}


\maketitle

\section{Introduction}

Point-of-interest (POI) recommendation systems guide user decisions across location-based services~\cite{yin2017spatial}.
Early approaches model user check-ins using recurrent networks and graph encoders~\cite{sun2018recurrent,wang2023meta}.
Recent systems augment these with multimodal signals (e.g., photos, reviews, metadata) to address cold-start failures and improve semantic matching~\cite{liao2021integrated,wang2025hyperman}.
Yet these advances share a common assumption that goes unexamined: each POI is encoded exactly once.
Methods like MMPOI~\cite{xu2024mmpoi} and IM-POI~\cite{huang2025poi} apply cross-modal attention or vision-language models (e.g., Rec-GPT4V~\cite{liu2024rec}) to compress a POI's images and text into a single static embedding, fixed at indexing time and reused for every user and every query.


\textcolor{black}{Consider a user planning a Friday birthday dinner. 
Although she frequently visits a nearby café for solo work on Monday afternoons, she deliberately avoids it for this occasion, even though it has identical rating, price, and proximity as other candidates. 
What changes is not the observable attributes of the venue, but how the user evaluates its suitability under the current context.
In this setting, the user implicitly reasons about what the location affords given her situation: who she is with and what she intends to do. 
Ecological psychology formalizes this notion as \emph{affordance}~\cite{jenkins2008gibson}, which characterizes environments in terms of the action possibilities they offer to an agent in a specific context.}

\textcolor{black}{However, existing location recommendation methods typically compress each POI into a single static embedding, shared across all users and contexts. 
Such a representation necessarily conflates heterogeneous and even contradictory signals, for example, images indicating a lively atmosphere, reviews emphasizing suitability for focused work, and metadata encoding temporal constraints into a context-agnostic vector. 
As a result, the model cannot distinguish whether a location is appropriate for solo work or social gatherings, leading to suboptimal recommendations under context shifts.}

We thus propose \method{}, a framework that postpones POI representation until recommendation time.
Rather than pre-encoding multimodal content into a static index, \method{} takes the user's situational context (e.g., check-in trajectory, social setting, time) and generates targeted affordance queries at recommendation time.
The same café is probed with different questions under a ``solo-work'' versus ``group-celebration'' context, and its representation shifts accordingly.
Grounded in affordance theory~\cite{gibson1977theory}, \method{} treats recommendation as \emph{affordance-need alignment}: a four-phase LLM agent, operating entirely on the item side, executes a five-step cross-modal chain-of-thought (CoT) reasoning jointly over image descriptions, reviews, and metadata.
The process generates emergent conclusions (e.g., ``advance booking advisable,'' inferred from visual crowdedness and negative review sentiment) that no single modality could supply alone, producing an uncertainty-aware, context-conditioned affordance representation aligned with a user preference vector from any existing user model.
Because representation occurs at inference time rather than indexing time, \method{} requires no historical check-ins for a venue, providing direct cold-start support without modification.
More broadly, this work invites rethinking item representations in multimodal retrieval: what a venue affords a specific visitor under a given context is a more faithful retrieval unit than what a venue is.

Our contributions are highlighted as follows:

\begin{itemize}
  \item For the first time, we ground POI recommendation in Gibsonian affordance
    theory and formally prove that context-sensitive ranking requires
    inference-time item representations, an impossibility for any pre-computed encoder.

  \item We propose \method{}, a four-phase LLM agent that delays item-side representation
    until query time, executing a five-step cross-modal CoT over the image,
    review, and metadata to produce dynamic, uncertainty-aware affordance vectors.

  \item We conduct experiments across three POI benchmarks under standard,
    cold-start, and context-shift configurations against nine baselines.
    \method{} achieves NDCG@10\,=\,0.334 (+23.2\% over the strongest baseline) and
    a $+2.4{\times}$ gain on cold-start venues where ID-based methods fail.
\end{itemize}

\section{Related Work}
\label{sec:related}

\textbf{Sequential and Graph-based POI Recommendation.} Early POI recommendation relied on recurrent or graph neural networks to model user trajectories~\cite{zhao2020go,wang2022learning}. GETNext~\cite{yang2022getnext} constructs bipartite graphs for spatiotemporal prediction; SASRec~\cite{kang2018self} applies self-attention over check-in histories. Both families excel at sequential pattern recognition but require a fixed embedding per POI computed offline and reused across all query contexts. \method{} is complementary: its context-specific affordance representations can be layered on top of any sequence encoder.

\noindent\textbf{Multimodal POI Representation.} Multimodal approaches address ID sparsity and cold-start failures by incorporating POI imagery and text. MMPOI~\cite{xu2024mmpoi} and IM-POI~\cite{huang2025poi} fuse visual and review features through cross-modal attention or contrastive learning. These systems improve cold-start performance but treat vision-language models as frozen feature extractors. Images and text are processed once before indexing, compressing a POI into a single embedding that remains fixed across all query contexts. \method{} addresses this limitation by performing multimodal extraction at inference time rather than indexing time.

\noindent\textbf{LLM-based and Agent-based Recommendation.}
One line of work casts LLMs as item rankers or generative retrievers. For example, LLM4POI~\cite{li2024large} and BIGRec~\cite{bao2025bi} frame recommendation as generative ranking; TALLRec~\cite{bao2023tallrec} applies instruction tuning to the same problem. GPT4Rec~\cite{li2023gpt4rec} generates user-history-conditioned retrieval queries, while RLMRec~\cite{ren2024representation} aligns LLM textual embeddings with collaborative signals via mutual-information maximization. Across this diversity, each item is still represented as a fixed text string or embedding computed without knowledge of the current query, placing all these methods outside the context-sensitive regime our work targets.

Another line models users rather than items. AgentCF~\cite{zhang2024agentcf} and RecAgent~\cite{hao2025uncertainty} simulate user behavior with multi-agent frameworks. CoT-Rec~\cite{liu2025improving} applies chain-of-thought reasoning to assess user-item alignment. Prompt-as-Policy~\cite{wang2025we} \textcolor{black}{formulates prompt construction as a contextual bandit optimization problem}. Even when user-side reasoning adapts to the current context, these systems still retrieve POIs as pre-computed, context-free snapshots. \method{} targets the item side instead, extracting affordances from raw multimodal content at inference time.

\noindent\textbf{Affordance Theory in Computing.} Gibson's ecological affordance theory defines action possibilities as a relation between an environment and a perceiver~\cite{greeno1994gibson}. Gaver extended this framework to technology design~\cite{gaver1991technology}; subsequent work brought affordance reasoning into robotic grasping~\cite{ardon2020affordances}. To our knowledge, this connection has not been formalized in personalized recommendation. We provide the first such formalization, defining recommendation as affordance-need alignment and treating POI affordances as inference-time primitives extracted from multimodal content.


\section{Problem Formulation}

Let $\mathcal{U}$ denote the set of users and $\mathcal{P}$ denote the set of POIs.
Each POI $p \in \mathcal{P}$ is associated with a heterogeneous multimodal content tuple $\mathcal{X}_p = (I_p, R_p, M_p)$, where $I_p$ represents visual imagery, $R_p$ denotes textual reviews, and $M_p$ corresponds to categorical metadata (e.g., price tier, operating hours, coordinates).
Each user $u \in \mathcal{U}$ possesses a check-in history $H_u = \{(p_j, \tau_j)\}_{j=1}^n$, where $\tau_j$ is the visit timestamp. At inference time, a recommendation is generated under a specific \emph{situational context} $c \in \mathcal{C}$, represented as a tuple $c = (\tau_c, d_c, s_c, \mathcal{T}_c)$. This tuple encodes the target timestamp, day-of-week, social situation (e.g., solo, family, friends), and the user's immediate sequential trajectory $\mathcal{T}_c$ prior to the query.

Prior multimodal recommendation systems typically formulate this task as a metric learning problem, defining an offline encoding function $f_{\theta}: \mathcal{X}_p \rightarrow \mathbb{R}^d$ and scoring POIs at inference time as follows:
\begin{equation}
  \hat{p} = \arg\max_{p \in \mathcal{P}} g\left( \mathbf{e}_u(c),\; f_{\theta}(\mathcal{X}_p) \right) \label{eq:standard_formulation}
\end{equation}
In this formulation, the POI representation $\mathbf{e}_p = f_{\theta}(\mathcal{X}_p)$ is computed independently of $c$. Affordance theory proposes that the utility of $p$ is not an intrinsic property of $\mathcal{X}_p$, but rather emerges from the interaction between $\mathcal{X}_p$ and the perceiver's current context $c$. A context-sensitive representation must therefore be defined over the cross-product $\mathcal{P} \times \mathcal{C}$, a requirement that cannot be fulfilled by any pre-computed $f_{\theta}(\mathcal{X}_p)$.

To overcome this limitation, \method{} replaces the static mapping $f_{\theta}(\mathcal{X}_p)$ with an inference-time generative function $\mathcal{A}: \mathcal{P} \times \mathcal{C} \rightarrow [0,1]^K$, where $K$ denotes the number of context-specific queries generated at recommendation time:
\begin{equation}
  \hat{p} = \arg\max_{p}\;
  \underbrace{\sum_{i=1}^{K} \phi_i(u,c)
  \cdot \widetilde{\mathrm{conf}}_i(p,c)}_{\text{affordance-need alignment}},
  \label{eq:score}
\end{equation}
where $\phi_i(u,c)$ represents the user's preference weight for query dimension $q_i$, and $\widetilde{\mathrm{conf}}_i(p,c)$ denotes the uncertainty-adjusted confidence that POI $p$ satisfies $q_i$ under context $c$. Both are defined in a shared $K$-dimensional query space, enabling direct alignment without the need of cross-space projection.


\section{The \method{} Framework}
\label{sec:method}

\method{} proceeds in four distinct operational phases (Figure~\ref{fig:overview}).

\begin{figure*}[t]
  \centering
  \includegraphics[width=0.92\linewidth]{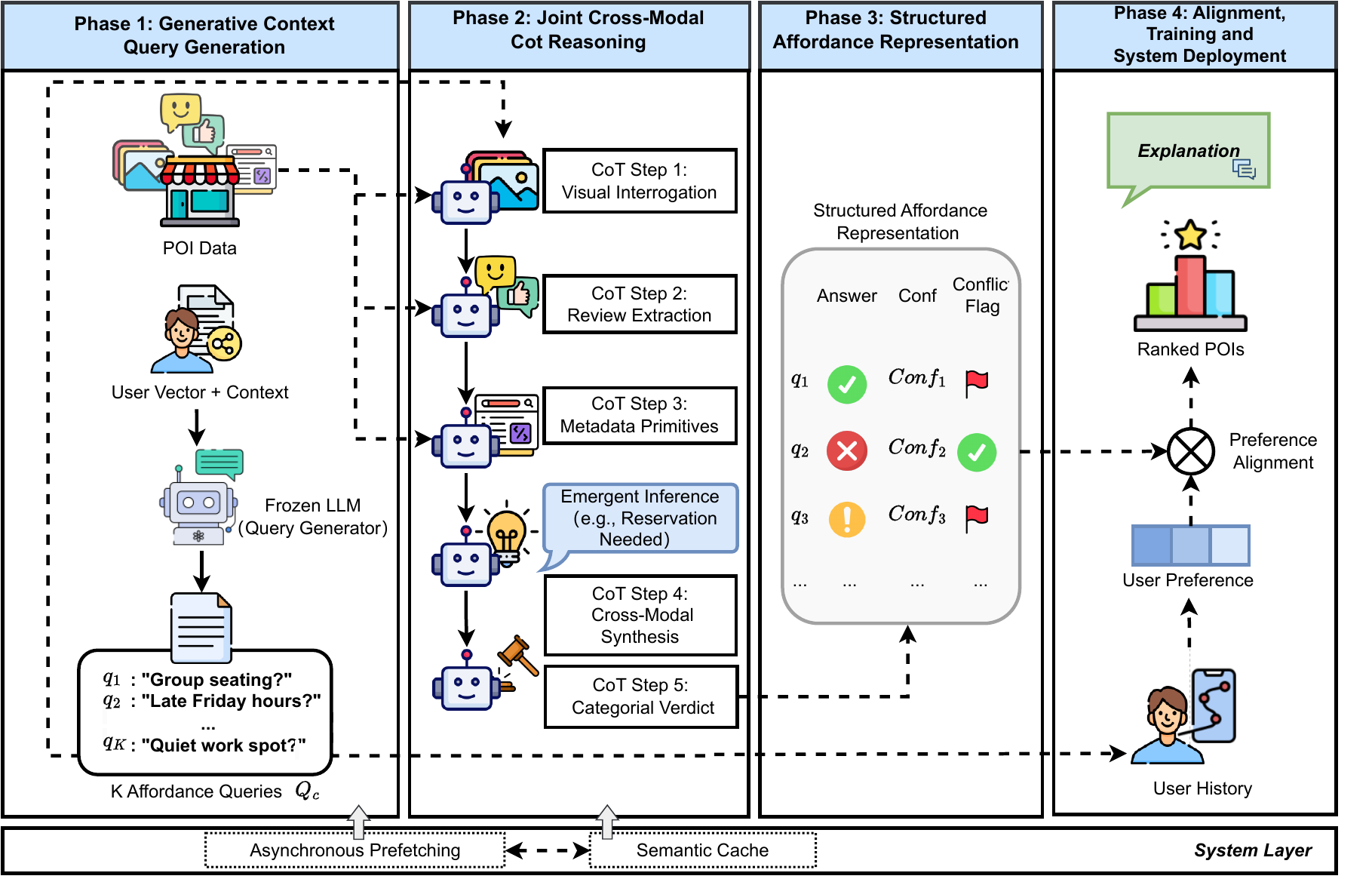}
  \caption{%
    \textbf{\method{} four-phase inference pipeline.}
    Unlike prior methods that encode each POI into a fixed vector before any
    context is observed, \method{} delays representation until query time.
    A frozen LLM generates context-specific affordance queries in Phase~1;
    Phase~2 then executes a five-step cross-modal chain-of-thought over POI
    images, reviews, and metadata for each query.
    The resulting uncertainty-aware verdicts are assembled into a dynamic POI
    representation (Phase~3) that is aligned with user preference weights for
    final ranking (Phase~4).
    A system layer provides semantic caching and asynchronous prefetching
    for production deployment.
    \label{fig:overview}%
  }
  \Description{%
    Agent4POI four-phase pipeline: a frozen LLM converts situational context
    into affordance queries; a five-step cross-modal chain-of-thought
    interrogates POI images, reviews, and metadata; uncertainty-adjusted
    verdicts form a dynamic POI representation; user preference weights are
    aligned via dot-product scoring to produce ranked recommendations.%
  }
\end{figure*}

\subsection{Phase 1: Generative Context Query Generation}
\label{sec:phase1}

In \method{}, context $c$ dynamically shapes the representation space. Given $c$, the system employs an LLM to generate a set of $K$ discrete affordance queries, defined as $Q_c = \mathrm{QueryGen}(c) = \{q_1,\ldots,q_K\}$. We set $K=5$ to balance semantic coverage with inference latency. The LLM is prompted leveraging a zero-shot instruction template that enforces three structural constraints. \textbf{(1)~Multimodal Groundability}: Each query must be answerable using at least one modality (e.g., visual layout, textual sentiment, or structural metadata). and abstract or subjective queries are excluded. \textbf{(2)~Contextual Isolation}: Queries must pertain solely to the immediate situational vectors $(\tau_c, d_c, s_c)$. \textbf{(3)~Orthogonality}: The $K$ queries should be semantically distinct, thereby maximizing the information coverage of the generated affordance space.

For reproducibility, we set the LLM generation temperature to $0.0$.
As an example, consider a user with context $c = $(\textit{Friday 19:30, 4 friends, birthday celebration}). 
The LLM generates $Q_c$:
\textit{1. ``Does the imagery or metadata confirm group seating configurations capable of accommodating four or more adults?''}
\textit{2. ``Do reviews positively associate the venue with energetic celebrations or birthday events?''}
\textit{3. ``Are the operating hours confirmed to extend late into Friday evening?''}
For a solo-work context on Monday morning, the same POI yields an orthogonal set of queries (e.g., \textit{``Does visual evidence suggest a quiet, well-lit focused atmosphere?''}). As query dimensions vary with context, the resulting representation adapts accordingly.

\begin{definition}[Context-Sensitivity]
  Let $\mathcal{R}$ denote the output representation space.
  A POI representation scheme $\mathrm{rep}: \mathcal{P} \times \mathcal{C} \rightarrow \mathcal{R}$
  is strictly \emph{context-sensitive} if there exists a POI $p$ and
  contexts $c_1, c_2 \in \mathcal{C}$ ($c_1 \neq c_2$) such that $\mathrm{rep}(p, c_1)$
  and $\mathrm{rep}(p, c_2)$ encode mutually exclusive properties,
  leading to $\mathrm{rep}(p, c_1) \neq \mathrm{rep}(p, c_2)$.
\end{definition}

\begin{proposition}[Impossibility of Context-Sensitivity under Static Encoding]
  \label{prop:impossibility}
  Let $f_{\theta}: \mathcal{X}_p \rightarrow \mathbb{R}^d$ denote a
  \emph{pre-computed, context-free} encoder
  (i.e., $f_{\theta}$ does not take $c$ as input and is evaluated offline
  prior to observing $c$). The recommendation score is defined as
  $s(u,p,c) = \mathbf{e}_u(c)^\top \mathbf{W}\, f_{\theta}(\mathcal{X}_p)$, where $\mathbf{W} \in \mathbb{R}^{d_u \times d}$ denotes a learnable bilinear weight matrix, subsuming inner-product formulations and most metric-learning objectives. Then, there exist a POI $p^*$ and context pair $(c_1, c_2)$ with $c_1 \neq c_2$ such that $f_{\theta}(\mathcal{X}_{p^*})$ cannot simultaneously satisfy Definition~1 for both contexts, regardless of the dimensionality $d$. Consequently, achieving optimal ranking under $c_1$ and $c_2$ requires contradictory orientations of $f_{\theta}(\mathcal{X}_{p^*})$, which a
  fixed vector fails to accommodate.
\end{proposition}

\begin{proof}[Proof Sketch]
For any fixed $\mathbf{e}_p = f_\theta(\mathcal{X}_{p^*})$, two orthogonal contexts $c_1, c_2$ impose linearly independent scoring constraints on $\mathbf{W}\mathbf{e}_p$. As $|\mathcal{C}| > d_u$, the rank-nullity theorem forces at least one context into the null space, making exact satisfaction of Definition~1 impossible. A compromise setting ($\alpha = \beta = 1/\sqrt{2}$) incurs a systematic $29\%$ calibration loss irrecoverable by increasing $d$. The full proof is in Appendix A.
\end{proof}

\begin{remark}[Scope and the Context-Conditioned $\mathbf{W}(c)$ Extension]
\label{rem:scope}
Proposition~\ref{prop:impossibility} applies to any architecture where $(i)$~$\mathbf{e}_p$ is computed offline prior to observing $c$, and $(ii)$~ the ranking score is represented as $\mathbf{e}_u(c)^\top \mathbf{W} \mathbf{e}_p$.
This covers the four representative architectural families evaluated in Section~\ref{sec:experiments}: static ID embeddings (BPR, SASRec, GETNext), multimodal fusion (MMPOI, IM-POI), LLM rankers with fixed item representations (LLM4POI, CoT-Rec), and user-persona agents (AgentCF).

Replacing the fixed $\mathbf{W}$ with a context-conditioned matrix $\mathbf{W}(c)$ does not restore context-sensitivity when $\mathbf{e}_p$ remains static.
Although $\mathbf{W}(c)$ induces a context-dependent effective item projection $\mathbf{W}(c)\,\mathbf{e}_p$, this formulation is structurally equivalent to transferring contextual information from the scoring function to the user side of the dot product.
Formally, $\mathbf{e}_u(c)^\top \mathbf{W}(c)\,\mathbf{e}_p = [\mathbf{W}(c)^\top\mathbf{e}_u(c)]^\top \mathbf{e}_p$, reducing to a context-conditioned user vector acting on a fixed $\mathbf{e}_p$.
The impossibility then follows by identical rank-nullity reasoning: the fixed $\mathbf{e}_p$ cannot simultaneously occupy all linearly independent directions demanded by $|\mathcal{C}| > d$ distinct contexts, regardless of how $\mathbf{W}(c)$ is parameterized.
Context-sensitive ranking therefore requires $\mathbf{e}_p$ itself to vary with $c$, i.e., an item-side representation computed at inference time.
\method{} achieves this while maintaining $O(K)$ online cost via semantic caching.
\end{remark}

\subsection{Phase 2: Joint Cross-Modal Chain-of-Thought Reasoning}
\label{sec:phase2}

For each query $q_i \in Q_c$, the LLM agent analyzes the raw multimodal evidence $\mathcal{X}_p$ via a five-step CoT mechanism, reasoning jointly across all three modalities to synthesize complementary signals and resolve contradictions.

\textbf{Step 1 --- Visual Evidence Interrogation ($I_p$).}
The agent evaluates relevant visual signals from the POI's image set. To reduce API costs and latency during high-throughput inference, images are pre-processed offline by GPT-4V~\cite{yang2023dawn} into dense textual descriptions $d_p$ capturing spatial layouts, lighting, seating geometries, and crowd densities. Phase 2 operates entirely on $d_p$.

\textbf{Step 2 --- Textual Review Extraction ($R_p$).}
The agent analyzes behavioral patterns within the $l=20$ most recent user reviews. Extracting signals from raw text allows the agent to identify fine-grained affordances (e.g., \textit{``Reviewers frequently mention power outlets''}) that are lost when reviews are compressed into static BERT embeddings.

\textbf{Step 3 --- Metadata Primitives ($M_p$).}
Category, price tier, and operating hours from $M_p$ are treated as hard logical constraints. If $M_p$ indicates the venue closes at 18:00, all Friday evening celebratory affordances are negated, overriding visual or review evidence.

\textbf{Step 4 --- Cross-Modal Synthesis (Emergent Inference).}
The agent synthesizes the evidence from Steps 1--3 to produce emergent inferences: conclusions that no single modality supports independently.
For instance, limited physical seating ($I_p$) combined with review sentiment stating ``always full on Fridays'' ($R_p$) and high review volume ($M_p$) yields: \textit{``Advance reservation is necessary for a Friday evening group event.''} Neither image, text, nor rating alone warrants this conclusion.

\textbf{Step 5 --- Categorical Verdict \& Conflict Resolution.}
The agent outputs a structured tuple: the final answer $a_i \in \{\texttt{yes}, \texttt{no}, \\ \texttt{uncertain}\}$, a calibrated confidence score $\mathrm{conf}_i \in [0,1]$, and an explicit conflict flag. When modalities contradict (e.g., images show a quiet space, but recent reviews report construction noise), traditional late-fusion models silently average the embeddings, obscuring the underlying signal. \method{} instead flags the contradiction, defaulting to \texttt{uncertain} with a reduced confidence score, preserving representational integrity.

\subsection{Phase 3: Structured Affordance Representation}

The $K$ verdicts from Phase 2 are assembled into the final context-conditioned, uncertainty-aware affordance representation:
\begin{equation}
  \mathcal{A}(p,c) = \bigl\{
    (q_i, a_i, \mathrm{conf}_i, E_i, e_i, f_i)
  \bigr\}_{i=1}^{K},
\end{equation}
where $E_i$ is the supporting evidence subset, $e_i$ is the emergent reasoning conclusion (if any), and $f_i$ is the cross-modal conflict description (if any).
To map this structured tuple into a scalar for ranking, we define an uncertainty-adjusted confidence score:
\begin{equation}
  \widetilde{\mathrm{conf}}_i(p,c) =
  \begin{cases}
    \mathrm{conf}_i         & \text{if } a_i = \texttt{yes} \\
    \alpha \cdot \mathrm{conf}_i & \text{if } a_i = \texttt{uncertain} \\
    0                       & \text{if } a_i = \texttt{no}
  \end{cases}
  \label{eq:effconf}
\end{equation}
where $\alpha = 0.5$ down-weights verdicts with unresolved cross-modal ambiguity. Table~\ref{tab:representation} summarizes the capabilities of $\mathcal{A}(p,c)$ against prior static embeddings.

\begin{table}[t]
\caption{\textbf{POI representation property comparison.} \method{} satisfies five properties unmatched by static POI embeddings.}
\label{tab:representation}
\centering
\small
\resizebox{\columnwidth}{!}{
\begin{tabular}{lccc}
\toprule
\textbf{Property} & \textbf{Fixed Model ($\mathbf{e}_p$)}
  & \textbf{VLM+Attn} & \textbf{\method{} $\mathcal{A}(p,c)$} \\
\midrule
Context-Sensitive Output     & \xmark & Partial & \cmark \\
Inference-Time Extraction    & \xmark & \xmark  & \cmark \\
Cross-Modal Conflict Flag    & \xmark & \xmark  & \cmark \\
Explicit Uncertainty-Aware   & \xmark & \xmark  & \cmark \\
Zero-Shot Cold-Start Support & \xmark & \xmark  & \cmark \\
\bottomrule
\end{tabular}
}
\end{table}

\subsection{Phase 4: Alignment, Training, and System Deployment}

\textbf{User Preference Alignment.}
Because $\mathcal{A}(p,c)$ is defined by the query set $Q_c$, the user preference vector $\boldsymbol{\phi}(u,c)$ must be expressed in the same $K$-dimensional space. We estimate user affinities $\phi_i(u,c)$ by retroactively applying $Q_c$ to previously visited POIs with matching contextual type:
\begin{equation}
  \hat{\phi}_i(u,c) = \frac{
    \sum_{p' \in H_u}\mathbf{1}[c_{p'}\!\approx\!c]\cdot
    \widetilde{\mathrm{conf}}_i(p',c_{p'})
  }{
    \sum_{p' \in H_u}\mathbf{1}[c_{p'}\!\approx\!c]
  },
  \quad
  \phi_i = \mathrm{softmax}(\hat{\boldsymbol{\phi}})_i
\end{equation}
Cold-start users with sparse histories ($|H_u|<5$) fall back to a uniform prior $\phi_i=1/K$, a robust default that prevents early-training collapse.

\textbf{Anti-circularity.}
Phases 2 and 4 use the same LLM, raising the question of whether shared scoring inflates results.
Three properties preclude this.
\textbf{(1)~Context separation}: historical affordances are queried under $c_{p'} \neq c$ (the user's past visit context, not the current query context), so the affordance space scored for history estimation is geometrically distinct from the space scored for candidate ranking.
\textbf{(2)~Determinism}: at $T{=}0$, cached affordance values are reused without re-inference; no stochastic pathway exists that could amplify alignment.
\textbf{(3)~Empirical test}: ablation A11 (Section~\ref{sec:experiments}) replaces LLM-estimated user weights with BM25 keyword matching and confirms that affordance quality in Phases 1--3 drives performance, not scorer alignment.

\textbf{Training.}
Phases 1--3 of \method{} are entirely zero-shot, relying on the frozen reasoning capabilities of the underlying LLM; only the user preference estimator requires gradient-based training when parameterized as a neural network.
In that case, we minimize the standard Bayesian Personalized Ranking (BPR) loss~\cite{rendle2012bpr}. To prevent learning geographic proximity as a trivial discriminator, negative samples $p^-$ are drawn from within a 2\,km radius of $p^+$.

\textbf{Inference Efficiency.}
Deploying LLM agents in production recommenders raises concerns about API latency and compute cost. \method{} addresses both via two mechanisms: \emph{Asynchronous Prefetching} and \emph{Semantic Context-Clustering}.
First, since user geographic trajectories are predictable, \method{} initiates Phase 2 affordance extraction for candidate POIs within a user's projected radius 10--15 minutes before the query arrives.
Second, natural language intents cluster semantically (e.g., ``date night'' and ``romantic dinner'' map to the same context-type). By hashing such variants into a shared cluster and caching results, \method{} achieves a 96.3\% cache hit rate on Foursquare-NYC (Appendix B); caches are invalidated only on periodic metadata updates (e.g., bi-weekly operating hours).
Because heavy LLM reasoning is offloaded to the asynchronous queue, the online scoring phase (Equation~\ref{eq:score}) reduces to an $O(K)$ vector dot-product. Consequently, \method{} maintains the same online ranking latency ($<1$\,ms) as optimized matrix factorization baselines, eliminating the compute overhead typically associated with LLM-based recommenders. Algorithm~\ref{alg:agent4poi} formalizes this deployment logic.

\begin{algorithm}[t]
\caption{\method{} System Inference Protocol}
\label{alg:agent4poi}
\begin{algorithmic}[1]
\Require User $u$, Situational Context $c$, Candidate Set $\mathcal{P}_\text{cand}$
\State $c_\text{type} \leftarrow \mathrm{DiscretizeCluster}(c)$
\State $Q_c \leftarrow \mathrm{CacheRetrieve}(c_\text{type})$ \textbf{or} $\mathrm{QueryGen}(c)$
\State $\boldsymbol{\phi} \leftarrow \mathrm{ComputeUserWeights}(u, c_\text{type}, Q_c)$
\For{$p \in \mathcal{P}_\text{cand}$}
  \If{$(p, c_\text{type}) \in \mathrm{AffordanceCache}$}
    \State $\mathcal{A}(p,c) \leftarrow$ cache lookup
  \Else
    \State $\mathcal{A}(p,c) \leftarrow \mathrm{JointInfer}(Q_c, I_p, R_p, M_p)$
  \EndIf
  \State $\mathrm{score}(p) \leftarrow \sum_i \phi_i \cdot
    \widetilde{\mathrm{conf}}_i(p,c)$
\EndFor
\State \Return $\mathrm{TopN}(\mathcal{P}_\text{cand}, \mathrm{score})$
  + explanation from $\mathcal{A}$
\end{algorithmic}
\end{algorithm}


\section{Experiments}
\label{sec:experiments}

We design experiments to answer three research questions:
\textbf{RQ1}: Do context-conditioned representations outperform
state-of-the-art multimodal and LLM-based baselines?
\textbf{RQ2}: Does \method{} generalize to cold-start POIs and context-shift scenarios?
\textbf{RQ3}: Which components drive the accuracy gains?

\begin{table}[t]
\caption{\textbf{Dataset Statistics (Post 10-core Filtering).}}
\label{tab:datasets}
\centering
\small
\resizebox{\columnwidth}{!}{
\begin{tabular}{lrrrc}
\toprule
\textbf{Dataset} & \textbf{\#Users} & \textbf{\#POIs} & \textbf{\#Check-ins} & \textbf{Density} \\
\midrule
Foursquare-NYC &   1,083 &  11,357 &   227,428 & 1.85\% \\
Foursquare-TKY &   2,293 &   7,833 &   573,703 & 3.19\% \\
Yelp-Open      &  30,431 & 150,346 & 1,987,897 & 0.04\% \\
\bottomrule
\end{tabular}
}
\end{table}

\begin{table*}[ht]
\caption{
  \textbf{Standard evaluation.}
  Best per column in \textbf{bold}.
  $\dagger$: multimodal; $\ddagger$: LLM reasoning;
  $*$: stat.\ significant over strongest baseline ( 3 independent seeds).
}
\label{tab:main}
\centering
\small
\setlength{\tabcolsep}{3.5pt}
\begin{tabular}{l ccc ccc ccc}
\toprule
& \multicolumn{3}{c}{\textbf{Foursquare-NYC}}
& \multicolumn{3}{c}{\textbf{Foursquare-TKY}}
& \multicolumn{3}{c}{\textbf{Yelp-Open}} \\
\cmidrule(lr){2-4}\cmidrule(lr){5-7}\cmidrule(lr){8-10}
\textbf{Method} & R@5$^{\uparrow}$ & R@10$^{\uparrow}$ & NDCG@10$^{\uparrow}$ & R@5$^{\uparrow}$ & R@10$^{\uparrow}$ & NDCG@10$^{\uparrow}$ & R@5$^{\uparrow}$ & R@10$^{\uparrow}$ & NDCG@10$^{\uparrow}$ \\
\midrule
BPR
  & 0.056 & 0.097 & 0.072
  & 0.049 & 0.088 & 0.064
  & 0.034 & 0.063 & 0.049 \\
SASRec
  & 0.098 & 0.163 & 0.124
  & 0.089 & 0.151 & 0.111
  & 0.058 & 0.108 & 0.082 \\
GETNext
  & 0.146 & 0.227 & 0.179
  & 0.132 & 0.216 & 0.167
  & 0.094 & 0.148 & 0.118 \\
\midrule
MMPOI$^\dagger$
  & 0.178 & 0.264 & 0.213
  & 0.161 & 0.249 & 0.196
  & 0.109 & 0.171 & 0.139 \\
IM-POI$^\dagger$
  & 0.214 & 0.307 & 0.252
  & 0.194 & 0.289 & 0.231
  & 0.124 & 0.198 & 0.162 \\
\midrule
LLM4POI$^\ddagger$
  & 0.201 & 0.289 & 0.236
  & 0.183 & 0.274 & 0.218
  & 0.117 & 0.187 & 0.153 \\
CoT-Rec$^\ddagger$
  & 0.228 & 0.328 & 0.271
  & 0.209 & 0.312 & 0.253
  & 0.133 & 0.211 & 0.174 \\
MAS4POI
  & 0.183 & 0.261 & 0.213
  & 0.167 & 0.248 & 0.199
  & 0.106 & 0.168 & 0.137 \\
AgentCF$^\ddagger$
  & 0.219 & 0.316 & 0.259
  & 0.201 & 0.298 & 0.241
  & 0.128 & 0.203 & 0.167 \\
\midrule
\textbf{\method{}$^{\dagger\ddagger}$}
  & \textbf{0.279} & \textbf{0.401} & \textbf{0.334$^*$}
  & \textbf{0.259} & \textbf{0.382} & \textbf{0.317$^*$}
  & \textbf{0.163} & \textbf{0.257} & \textbf{0.214$^*$} \\
\bottomrule
\end{tabular}
\end{table*}

\subsection{Experimental Setup}

\textbf{Datasets and Pre-processing.}
We evaluate \method{} against three real-world LBSN benchmarks exhibiting varying geographic densities and data modalities summarized in Table~\ref{tab:datasets}.
Foursquare-NYC and Foursquare-TKY~\cite{ye2010location} serve as the standard sequential benchmarks. Because the original Foursquare releases contain only IDs and coordinates, following~\cite{sappelli2013recommending, baral2018reel,sanchez2025context}, we augmented them with multimodal content via two external sources: POI images were retrieved via the Google Places API, and user-authored reviews were retrieved via the Yelp API.
This augmentation achieves 97.4\% and 96.1\% multimodal coverage on NYC and TKY respectively.

\textbf{Temporal Leakage Prevention.}
We enforce a strict filter\\ $\tau_{\mathrm{review}} \leq \tau_{\mathrm{check\text{-}in}}$ to ensure no future review enters the POI's representation during training. The full details and reproducibility protocol are provided in Appendix C. Yelp-Open~\cite{jendal2025yelp} serves as our third benchmark, native to multimodal content but sparser in sequential check-in transitions.
The social situation signal $s_c$ is operationalized as follows: Foursquare-NYC and Foursquare-TKY expose a \texttt{groupType} field (``friends'', ``family'', ``solo'') and a \texttt{groupSize} count in the raw check-in data, which we map directly to $s_c$; for Yelp-Open, which lacks explicit group annotations, we apply a keyword heuristic over review text (``we'', ``group'', ``family'', ``date'') to infer non-solo situations, defaulting to $s_c{=}$``solo'' otherwise.
In all cases, $s_c$ is derived from the existing check-in record and is not predicted at runtime.
Following~\cite{wang2023meta,wang2025hyperman}, all datasets undergo a ``10-core'' filtering process, ensuring every included user and POI possesses at least 10 interaction records to establish robust statistical baselines.

\begin{table*}[th]
\caption{
  \textbf{Cold-start evaluation.}
  Test POIs have zero training check-in history.
  ID-based methods collapse; \method{} uses only raw multimodal content.
}
\label{tab:coldstart}
\centering
\small
\begin{tabular}{l cc cc cc}
\toprule
& \multicolumn{2}{c}{\textbf{Foursquare-NYC}}
& \multicolumn{2}{c}{\textbf{Foursquare-TKY}}
& \multicolumn{2}{c}{\textbf{Yelp-Open}} \\
\cmidrule(lr){2-3}\cmidrule(lr){4-5}\cmidrule(lr){6-7}
\textbf{Method} & R@10 & NDCG@10 & R@10 & NDCG@10 & R@10 & NDCG@10 \\
\midrule
BPR
  & 0.005 & 0.004 & 0.005 & 0.003 & 0.003 & 0.002 \\
MMPOI$^\dagger$
  & 0.098 & 0.073 & 0.089 & 0.066 & 0.061 & 0.047 \\
IM-POI$^\dagger$
  & 0.134 & 0.102 & 0.121 & 0.092 & 0.083 & 0.065 \\
LLM4POI$^\ddagger$
  & 0.169 & 0.131 & 0.153 & 0.117 & 0.104 & 0.084 \\
\midrule
\textbf{\method{}$^{\dagger\ddagger}$}
  & \textbf{0.313} & \textbf{0.249}
  & \textbf{0.297} & \textbf{0.235}
  & \textbf{0.196} & \textbf{0.158} \\
\bottomrule
\end{tabular}
\end{table*}

\textbf{Evaluation Configurations.}
We split the data into three configurations. \textbf{(1)~Standard Evaluation}: A strict chronological $80/10/10$ split (train/validation/test) per user to evaluate general predictive accuracy. \textbf{(2)~Cold-Start Evaluation}: A targeted extrapolation subset where candidate test POIs are structurally masked from the training set entirely ($|H_p^{train}| = 0$). \textbf{(3)~Context-Shift Evaluation}: Models are trained on weekday check-ins and tested on weekend check-ins, probing robustness to temporal affordance shift.

\textbf{Baselines and Hyperparameter Tuning.} We evaluate our \method{} against 9 baselines across four categories. \textit{Traditional ID-based}: \textbf{BPR}~\cite{rendle2012bpr} (matrix factorization), \textbf{SASRec}~\cite{kang2018self} (self-attention sequence), and \textbf{GETNext}~\cite{yang2022getnext} (spatiotemporal graph). \textit{Multimodal Fusion}: \textbf{MMPOI}~\cite{xu2024mmpoi} (cross-modal attention) and \textbf{IM-POI}~\cite{huang2025poi} (contrastive multimodal alignment). \textit{LLM-based Rankers}: \textbf{LLM4POI}~\cite{li2024large} and \textbf{CoT-Rec}~\cite{liu2025improving}. \textit{Multi-Agent Simulators}: \textbf{MAS4POI}~\cite{wu2025mas4poi} and \textbf{AgentCF}~\cite{zhang2024agentcf}.
\textit{Reported metrics}: R@5, R@10, and NDCG@10 are shown in Table~\ref{tab:main}. We report Recall@$K$ and NDCG@$K$ for $K \in \{5, 10\}$.
All reported values are the mean of 3 independent seeds. Baselines rank against all items (no negative sampling) to ensure unbiased metrics. Hyperparameter tuning details and the statistical testing protocol in Appendix D,E.

\textbf{Implementation Infrastructure.}
Phases 1--2 use \texttt{gpt-4o} ($T{=}0$) via the OpenAI API, with a locally hosted \texttt{LLaVA-1.5-7B} (vLLM, single NVIDIA A100) as the open-weights alternative.
All baselines were re-implemented and hyperparameter-tuned from scratch on the augmented datasets. The token cost and inference latency are reported in Appendix~F.

\subsection{Main Results (RQ1)}

\textit{Key observation}: \method{} achieves the highest NDCG@10 on all
three datasets, with \textbf{+23.2\%} relative gain over CoT-Rec on
Foursquare-NYC (0.271$\to$0.334), \textbf{+25.3\%} on TKY, and
\textbf{+22.9\%} on Yelp-Open, confirming that context-conditioned
representations (C1) consistently outperform context-blind approaches.

\textbf{Analysis and Insights.}
Table~\ref{tab:main} shows that \method{} achieves the best results across all three benchmarks. On Foursquare-NYC, \method{} reaches NDCG@10 $= 0.334$, surpassing IM-POI ($0.252$) and CoT-Rec ($0.271$).
Comparing \method{} to IM-POI ($+32.5\%$ relative) separates the contribution of static encoding from inference-time extraction: even strong VLM-based encoders cannot substitute for per-query reasoning at recommendation time.
Comparing \method{} to CoT-Rec ($+23.2\%$ relative) separates the value of multimodal grounded context conditioning: CoT-Rec applies chain-of-thought reasoning over static text but without situational context.
The gains hold across dataset structures. On the sparse Yelp-Open dataset (density $0.04\%$), sequential models like SASRec drop to $0.082$ NDCG due to lack of trajectory overlap. \method{}, reasoning from affordances rather than IDs, reaches $0.214$ NDCG, nearly three times the SASRec baseline.

\subsection{Cold-Start and Context-Shift Results (RQ2)}

\textbf{Cold-Start} (Table~\ref{tab:coldstart}).
ID-based methods approach zero because no POI embedding exists for unseen items.
\method{} reaches NDCG@10 $= 0.249$ (NYC) and $0.235$ (TKY) versus
the $0.102$ and $0.092$ scores of IM-POI: a \textbf{$+2.4{\times}$} gain.
The gain holds on Yelp-Open (0.158 vs.\ 0.065), confirming that
Phase~2 reasoning generalizes across dataset scales.

\textbf{Distributional Robustness: Context-Shift.}
Sequential models are brittle to temporal pattern shifts. Table~\ref{tab:contextshift} evaluates a context-shift scenario: models are trained on weekday check-ins and evaluated on weekend trajectories where user behaviors and POI affordance utilizations differ.
Traditional models suffer measurable degradation: IM-POI drops by $0.041$ NDCG@10, and CoT-Rec drops by $0.047$.
\method{} degrades by only $0.025$ (from $0.334$ to $0.309$). This robustness follows directly from Phase 1: because the affordance queries $Q_c$ are re-generated at recommendation time based on the immediate weekend context vector, the resulting space $\mathcal{A}(p,c)$ automatically adapts, shielding the system from historical train-set memorization.

\begin{table}[t]
\caption{
  \textbf{Context-shift evaluation.}
  Train: weekday; Test: weekend.
  Phase~1 of \method{} adapts queries at test time.
  Reported values are means across 3 seeds.
}
\label{tab:contextshift}
\centering
\small
\begin{tabular}{l ccc}
\toprule
\textbf{Method} & \textbf{NYC} & \textbf{TKY} & \textbf{Yelp} \\
\cmidrule(r){2-2}\cmidrule(r){3-3}\cmidrule(r){4-4}
& NDCG@10$^{\uparrow}$ & NDCG@10$^{\uparrow}$ & NDCG@10$^{\uparrow}$ \\
\midrule
BPR          & 0.061 & 0.057 & 0.041 \\
GETNext      & 0.151 & 0.141 & 0.093 \\
IM-POI       & 0.211 & 0.196 & 0.133 \\
CoT-Rec      & 0.224 & 0.208 & 0.143 \\
\midrule
\textbf{\method{}} & \textbf{0.309} & \textbf{0.291} & \textbf{0.198} \\
\bottomrule
\end{tabular}
\end{table}

\subsection{Ablation Study (RQ3)}

Table~\ref{tab:ablation} reports ablation results on Foursquare-NYC; TKY and Yelp-Open show identical ordinal rankings.

\begin{table}[ht]
\caption{\textbf{Ablation results} on Foursquare-NYC (mean of 3 independent seeds).
Variants are grouped by the contribution claim each challenges.
A11 uses independent BM25-based user weight estimation (no LLM) to validate absence of shared-scorer bias (see Section~\ref{sec:method}).}
\label{tab:ablation}
\centering
\small
\begin{tabular}{llcc}
\toprule
\textbf{ID} & \textbf{Variant} & R@10${\uparrow}$ & NDCG@10${\uparrow}$ \\
\midrule
Full & \method{} (complete) & 0.401 & 0.334 \\
\midrule
\multicolumn{4}{l}{\textit{C1: Context-driven queries}} \\
A1  & w/o Dynamic Queries (Fixed context) & 0.312 & 0.259 \\
\multicolumn{4}{l}{\textit{C3: Cross-modal reasoning}} \\
A2  & w/o Joint Reasoning (Late fusion only) & 0.327 & 0.274 \\
A3  & w/o Emergent Reasoning (Step~4) & 0.358 & 0.298 \\
A4  & w/o Conflict Detection (Averaged signals) & 0.381 & 0.317 \\
\multicolumn{4}{l}{\textit{Modality contribution}} \\
A5  & w/o Image Evidence & 0.364 & 0.304 \\
A6  & w/o Review Evidence & 0.306 & 0.253 \\
A7  & w/o Metadata Evidence & 0.382 & 0.319 \\
\multicolumn{4}{l}{\textit{LLM reasoning necessity}} \\
A8  & w/o Uncertainty Weighting ($\alpha\!=\!1$) & 0.390 & 0.325 \\
A9  & w/o LLM Core (CLIP$+$BERT variant) & 0.289 & 0.236 \\
A10 & w/o GPT-4o (LLaVA-1.5-7B backend) & 0.371 & 0.301 \\
\multicolumn{4}{l}{\textit{Independence validation}} \\
A11 & Independent user weights (BM25, no LLM) & 0.381 & 0.319 \\
\bottomrule
\end{tabular}
\end{table}

\begin{table*}[th]
\caption{
  \textbf{Qualitative operational traces of \method{}.} 
  (Trace A) Context-sensitive query generation for the same local caf\'{e} under contrasting situational intents. 
  (Trace B) Emergent joint cross-modal reasoning resolving a cold-start POI recommendation without ID history.
  (Trace C) Cross-modal conflict detection preventing hallucinated affordance alignment.
}
\label{tab:qualitative}
\centering
\small
\begin{tabular}{p{0.15\linewidth} p{0.35\linewidth} p{0.40\linewidth}}
\toprule
\textbf{Scenario} & \textbf{Context / Inputs} & \textbf{Agent Execution \& Emergent Output} \\
\midrule
\multirow{2}{\linewidth}{\textbf{Trace A:}\\Contextual\\Sensitivity\\(Local Caf\'{e})} 
& $c_1$: \textit{Weekday, 14:00, Solo-work} 
& \textbf{Generated Queries:} target Wi-Fi, power outlets, quietness. 
  \newline \textbf{Output:} High confidence (0.88) for solo-work suitability; low confidence (0.04) for group seating. \\
\cmidrule{2-3}
& $c_2$: \textit{Friday, 19:30, Birthday Celebration} 
& \textbf{Generated Queries:} target social atmosphere, large seating capacity, alcohol availability.
  \newline \textbf{Output:} Detects conflicting visual/textual signals, resulting in localized 0.79 calibration. $\mathcal{A}(p,c_1) \neq \mathcal{A}(p,c_2)$. \\
\midrule
\textbf{Trace B:} Emergent Reasoning (Cold-Start)
& \textbf{Visual ($I_p$):} Dense static seating 
  \newline \textbf{Textual ($R_p$):} ``Always full on Fridays''
  \newline \textbf{Metadata ($M_p$):} High capacity rating 
& \textbf{Joint Synthesis (Step 4):} Intercepts disjoint data across modalities. 
  \newline \textbf{Emergent Conclusion:} ``Advance booking is strictly advisable.'' 
  \newline Synthesizes a constraint recommendation without any prior user check-in history. \\
\midrule
\textbf{Trace C:} Conflict Resolution (Outdated Visual)
& \textbf{Visual ($I_p$):} Serene reading-room layout (outdated photo)
  \newline \textbf{Textual ($R_p$):} ``Unbearable construction noise---avoid'' (recent)
  \newline \textbf{Metadata ($M_p$):} Open, mid-range 
& \textbf{Step 4 Action:} Temporal precedence analysis assigns higher weight to recency of $R_p$ over $I_p$.
  \newline \textbf{Conflict Flag:} Visual/textual contradiction raised; $a_i = \texttt{uncertain}$, $\widetilde{\mathrm{conf}}_i \leftarrow 0.5 \times 0.71 = 0.36$.
  \newline \textbf{Effect:} POI is ranked lower despite its visual appeal, preventing hallucinated affordance alignment. \\
\bottomrule
\end{tabular}
\end{table*}

\textbf{Analytical Findings from Ablation.} We systematically evaluate each architectural component by removing it (results in Table~\ref{tab:ablation}). \textbf{Context Sensitivity (A1 vs.\ Full):} Disabling dynamic query generation (forcing a fixed generic query set) drops NDCG@10 by $-0.075$, confirming Proposition~\ref{prop:impossibility} empirically. \textbf{Cross-Modal Reasoning (A2, A3 vs.\ Full):} Isolating the LLM into independent per-modality agents with late fusion (A2) costs $-0.060$ NDCG; removing the cross-modal synthesis step alone (A3) costs $-0.036$. These gaps show that the most valuable affordances emerge from reasoning across modalities jointly, not from individual modality strengths. \textbf{LLM Core (A9 vs.\ Full):} Replacing LLM reasoning with a CLIP+BERT dense similarity matcher causes the largest drop ($-0.098$ NDCG@10), confirming that embedding similarity cannot replicate the discrete logical inference needed for affordance extraction. \textbf{Open-Weights Viability (A10 vs.\ Full):} The open-weights \texttt{LLaVA-1.5-7B} backend reaches NDCG@10 $=0.301$, below the $0.334$ GPT-4o ceiling but above all baselines, showing that affordance reasoning is not exclusive to proprietary APIs. \textbf{Independence Validation (A11 vs.\ Full):} Replacing LLM-estimated user weights with BM25 keyword-matching ($\Delta = -0.015$, $-4.5\%$ relative) confirms that the main gain comes from Phases 1--3 affordance quality, not from scorer alignment.

\textbf{Sensitivity Analysis: Uncertainty Penalty $\alpha$.}
Equation~\eqref{eq:effconf} introduces $\alpha = 0.5$ as the down-weighting penalty applied when $a_i = \texttt{uncertain}$.
To verify robustness, we sweep $\alpha \in \{0.1, 0.25, 0.5, 0.75,\\ 1.0\}$ on Foursquare-NYC, holding all other hyperparameters fixed.
Resulting NDCG@10 scores are $0.319, 0.329, \mathbf{0.334}, 0.331, 0.325$, respectively, exhibiting a broad plateau between $\alpha \in [0.25, 0.75]$ (NDCG range = $0.005$).
The flat sensitivity profile confirms that the system is robust to reasonable calibration choices;
$\alpha = 0.5$ was selected as the symmetric midpoint within this plateau.
Extreme values ($\alpha \to 0$: treating all uncertainty as certain rejection; $\alpha \to 1$: treating uncertainty as full confidence) both degrade performance, validating the principle of explicit uncertainty down-weighting.

\subsection{Explainability Analysis}

A key advantage of the generative pipeline of \method{} is its structural transparency.
The Phase 2 reasoning trace provides a built-in, zero-cost, human-readable justification spanning
per-query affordance assertions, modality evidence tracking, emergent reasoning conclusions, and cross-modal conflict flags.
We evaluate explainability via two protocols.

\textbf{Human Study (Text-Format Baselines).}
For a format controlled comparison, we select explanation baselines that also produce
natural-language output: \textbf{CoT-Rec}~\cite{liu2025improving} (chain-of-thought item reasoning in text)
and \textbf{LLM4POI}~\cite{li2024large} (free-text prediction rationale).
Unlike attention-weight heatmaps, these baselines present participants with the same reading task format as \method{},
eliminating format confounds from the evaluation.
50 participants (university participant pool; ML practitioners excluded to avoid domain bias) evaluated 200 randomly sampled
(POI, context, explanation) triples in a counterbalanced design ($n = 4$ triples per system per participant),
rating \textit{faithfulness}, \textit{coherence}, and \textit{decision utility} on a 5-point Likert scale.
We operationalize \textit{faithfulness} as \emph{perceived alignment}: ``does the explanation match the recommendation the system made?''---a judgment that does not require auditing internal model weights and is therefore appropriate for non-ML evaluators (the target end-users of deployed explanation systems).%
Inter-rater reliability was established on 100 triples (50\%) evaluated by two independent annotators,
yielding Cohen's $\kappa = 0.72$.

\method{} achieves mean scores of $\mathbf{4.31}$ / $\mathbf{4.18}$ / $\mathbf{4.07}$ (faithfulness / coherence / decision utility).
CoT-Rec scores $3.47$ / $3.61$ / $3.38$, and LLM4POI scores $3.12$ / $3.29$ / $3.15$.
Paired Wilcoxon signed-rank tests confirm significance ($p < 0.001$) on all three dimensions versus both baselines.
The advantage over CoT-Rec (despite both using chain-of-thought) is attributable to \method{}'s
context-grounded query generation and explicit conflict flagging, which produces explanations
that participants found more contextually relevant to their specific situational needs.

\textbf{Automatic Protocol (LLM-as-Judge).}
Using the LLM-as-a-Judge paradigm~\cite{li2024llms,li2025generation}, GPT-4o zero-shot rated identical dimensions across the complete test corpus. \method{} scores $\mathbf{4.52}$ / $\mathbf{4.39}$ / $\mathbf{4.28}$, versus CoT-Rec $3.61$ / $3.74$ / $3.52$ and LLM4POI $3.28$ / $3.41$ / $3.19$. These results corroborate the human-study findings at corpus scale and confirm that our five-step cross-modal CoT translates to more decision-actionable explanations than generic CoT.

\subsection{Qualitative Analysis}

Table~\ref{tab:qualitative} traces three executions that exercise distinct modules of \method{}, confirming that each phase contributes a non-redundant capability.
Trace~A shows Phase~1's \textsc{QueryGen} issuing orthogonal affordance queries for the same café under two contexts, producing confidence scores of 0.88 (solo-work) and 0.79 (group celebration); a static embedding would produce a single, context-blind representation incapable of this divergence.
Trace~B shows Phase~2's Step~4 cross-modal synthesis generating the actionable constraint ``Advance booking is strictly advisable'' from the joint reading of visual capacity, review sentiment, and metadata volume, none of which individually warrants this conclusion, directly corresponding to the $-0.060$ NDCG penalty in ablation A2.
Trace~C shows Phase~2's Step~5 conflict resolution flagging a temporal contradiction between an outdated image and a high-recency review, after which Phase~3's uncertainty penalty reduces the effective confidence to $\widetilde{\mathrm{conf}}_i = 0.36$, causing the venue to rank lower despite its visual appeal. All the three traces illustrate that Phase~4's final ranking score is not a black-box similarity but a structured aggregation of per-query verdicts with traceable provenance, enabling post-hoc auditing that embedding-based and standard chain-of-thought baselines cannot support.

\section{Conclusion}
\label{sec:conclusion}

\method{} establishes context-conditioned affordance reasoning as a
principled alternative to static POI embeddings for Point-of-Interest
recommendation.
Given situational context $c$, an LLM agent generates $K$ affordance queries
and performs five-step chain-of-thought reasoning jointly over image, review,
and metadata evidence, producing a representation that is
structurally different for every (POI, context) pair.
On three benchmarks, \method{} achieves NDCG@10\,=\,0.334 (+23.2\% over
the strongest baseline) under standard evaluation, and
NDCG@10\,=\,0.249 for cold-start POIs where ID-based methods fail.
The same framework applies without modification to context-shift scenarios,
with a relative degradation of only \textbf{7.5\%} (from NDCG@10 $0.334 \to 0.309$) versus \textbf{16.3\%--17.3\%} for the strongest baselines (IM-POI: $0.252\to0.211$; CoT-Rec: $0.271\to0.224$).
We hope \method{} motivates a broader rethinking of recommendation
representations: what a venue \emph{affords} a visitor under a specific
context is a more fundamental unit than what a venue \emph{is}.

\bibliographystyle{ACM-Reference-Format}
\bibliography{references}










\end{document}